\documentclass[letterpaper, 10 pt, conference]{ieeeconf}
\IEEEoverridecommandlockouts
\usepackage{balance}
\usepackage{color}
\usepackage[dvipsnames]{xcolor}
\usepackage{caption}
\usepackage{subcaption}
\usepackage{enumerate}
\usepackage{cite}
\usepackage{empheq}
\usepackage{mathrsfs}

\usepackage{mathtools}

\usepackage{tikz}
\usepackage{circuitikz}

\usepackage[font=small,labelfont=bf]{caption}

\makeatletter
\let\proof\@undefined                        
\let\endproof\@undefined                  
\makeatother
\usepackage{graphicx,amssymb,amstext,amsmath,amsthm}

\usepackage[bookmarks=true]{hyperref}
\usepackage{algorithm,algorithmicx,algpseudocode}
\algnewcommand{\algorithmicgoto}{\textbf{go to}}%
\algnewcommand{\Goto}[1]{\algorithmicgoto~\ref{#1}}%
\algnewcommand{\LineComment}[1]{\Statex \(\triangleright\) #1}
\algnewcommand{\LineCommentN}[1]{\Statex \hspace{1cm}\(\triangleright\) #1}

\usepackage{multirow}
\usepackage{stfloats}

\newtheorem{thm}{Theorem}
	
\newtheorem{lem}{Lemma}
\newtheorem{defn}{Definition}

\newtheorem{problem}{Problem}

\providecommand{\norm}[1]{\lVert#1\rVert}

\let\oldbibliography\thebibliography
\renewcommand{\thebibliography}[1]{%
  \oldbibliography{#1}%
}

\newcommand{\zhu}[1]{{\color{black} #1}}

\title{Time-Invariant Polytopic and Interval Observers for Uncertain Linear Systems via Non-Square  Transformation} 

\author{%
Feiya Zhu, Tarun Pati, 
and Sze Zheng Yong 
\thanks{
F. Zhu, T. Pati and S.Z. Yong are with Mechanical and Industrial Engineering Department, Northeastern University, Boston, MA 02115 (email: {\tt \{zhu.feiy,pati.ta,s.yong\}@northeastern.edu}). This work was supported in part by NSF grant CNS-2313814.}
\vspace{-0.55cm}
}

\allowdisplaybreaks
\begin{document}

\maketitle
\thispagestyle{empty}

\begin{abstract}
This paper presents novel polytopic and interval observer designs for uncertain linear continuous-time (CT) and discrete-time (DT) systems subjected to bounded disturbances and noise. Our approach guarantees enclosure of the true state and input-to-state stability (ISS) of the polytopic and interval set estimates. 
Notably, our approach applies to \emph{all} detectable systems that are stabilized by any optimal observer design, utilizing a potentially non-square (lifted) time-invariant coordinate transformation based on polyhedral Lyapunov functions and mixed-monotone embedding systems that do not impose any positivity constraints, enabling feasible and optimal 
observer designs, even in cases where previous methods fail. 
The effectiveness of our approach is demonstrated through several examples of uncertain linear CT and DT systems.
 \end{abstract}

  \vspace{-0.1cm}
\section{Introduction} 
 \vspace{-0.08cm}
State estimation is crucial for many engineering systems, including autonomous vehicles and power grids, enabling effective monitoring, decision-making, and control. Set-valued observers offer a robust approach, 
particularly when dealing with uncertain systems where the uncertainties are set-valued or non-stochastic, or when their distributions are unknown.

\noindent\emph{Literature Review.} 
For applications requiring strict accuracy bounds, e.g., safety-critical systems, set-valued estimation methods are preferred over stochastic approaches \cite{blanchini2012convex}. However,  characterizing the exact set of possible states is computationally prohibitive, even for simple linear systems with polytopic uncertainties \cite{shamma1999set}. Consequently, simpler geometric shapes, such as ellipsoids, intervals, or (constrained) zonotopes, are used to over-approximate these sets \cite{tahir2021synthesis,efimov2016design,briat2016interval,khajenejad_H_inf_2022,rego2020guaranteed,khajenejad2021guaranteed2,pati2022l₁,pati2025computationally}. 

Although polytopes provide tighter over-approximations of state sets compared to simpler convex shapes, the development of observers with polytopic estimates remains significantly under-explored. Current approaches, such as those in \cite{rego2020guaranteed, khajenejad2021guaranteed2}, primarily concentrate on guaranteed state estimation through reachability analysis and computational geometry, often neglecting observer gain design and stability analysis, and are predominantly limited to discrete-time systems. Notably, polytopic observers that incorporate both observer gain computation and stability analysis for both discrete and continuous-time systems are largely absent from the literature. Further, note that polytopic observers in this paper are to be distinguished from non-set-valued observers for polytopic systems with polytopic parameter variations.

On the other hand, interval observers have received significant attention for both linear and nonlinear systems \cite{tahir2021synthesis,efimov2016design,briat2016interval,pati2022l₁,khajenejad_H_inf_2022}. These observers typically rely on ensuring the error dynamics are simultaneously stable and positive/cooperative, often leading to complex design procedures. This has spurred the use of time-invariant or time-varying system transformations, e.g., \cite{chambon2016overview,mazenc2011interval,mazenc2014interval}, and additional degrees of freedom \cite{pati2022l₁,pati2025computationally} to broaden their applicability. However, as demonstrated in \cite{Mazenc2010,mazenc2014interval}, certain continuous and discrete-time systems with complex eigenvalues in their unobservable subsystems preclude time-invariant similarity transformations, motivating the development of time-varying approaches with time-varying transformations \cite{mazenc2011interval,mazenc2014interval} or time-varying reachability computations \cite{meslem_Hinf_2020,meslem2021reachability}. Despite these efforts, time-varying observers can be complex to implement, susceptible to numerical errors, and may yield conservative interval estimates. This motivates our exploration of potentially non-square time-invariant transformations to enable the design of stable, time-invariant polytopic and interval observers.

\noindent \emph{Contributions.} 
This paper addresses the open question of the existence of time-invariant set-valued observers for linear detectable systems by proposing a novel, potentially non-square, time-invariant coordinate transformation, drawing inspiration from infinity norm-based polyhedral Lyapunov functions for linear systems in\cite{polanski1995infinity,molchanov1986liapunov,christophersen2007further}. This transformation, which may involve state lifting \zhu{when it is} non-square, allows us to design feasible and optimal time-invariant polytopic and interval observers for any detectable linear system, overcoming the limitations of methods employing similarity transformations or additional degrees of freedom \cite{chambon2016overview,mazenc2011interval,mazenc2014interval,pati2022l₁,pati2025computationally}. In particular, we introduce a polytopic and interval design based on mixed-monotone embedding systems \cite{khajenejad2023tight}, which  are proven to be correct (i.e., the polytopic and interval set estimates enclose the true state) and whose set volumes are input-to-state stable (contracting \zhu{over time}), without imposing any additional positivity constraints. A significant corollary of our approach is the derivation of existence conditions for time-invariant interval observers (without state lifting), applicable to all detectable systems with real closed-loop eigenvalues and certain systems with complex eigenvalues. The efficacy of our method is demonstrated through comparative examples with time-varying approaches for both discrete- and continuous-time linear systems.
\vspace{-0.1cm}
 \section{Preliminaries}
 \vspace{-0.08cm}
 
 {\emph{Notation.} Let $ \mathbb{R}^n$, $\mathbb{R}^n_{>0}$, $\mathbb{R}^{n  \times p}$ denote the $n$-dimensional Euclidean space, the set of positive $n$-dimensional vectors, and the matrices of size $n$ by $p$, respectively. $\mathbb{N}_n$ represents natural numbers up to $n$, and  $\mathbb{N}$ denotes the natural numbers. For a vector $v \in \mathbb{R}^n$, its  $\infty$-norm is $\|v\|_{\infty}\triangleq \max_i |v_i|$. 
 For a matrix $M\in \mathbb{R}^{n  \times p}$, $M_{ij}$ refers to the entry in the $i$-th row and $j$-th column. The element-wise signum function of  $M$ is denoted by $\textstyle{\mathrm{sgn}}(M)$,  while $|M|\triangleq M^{\oplus}+M^{\ominus}$ is its element-wise absolute value, with $M^{\oplus}\triangleq \max(M,\mathbf{0}_{n\times p})$ and $M^{\ominus}\triangleq M^{\oplus}-M$. For a square matrix $M\in \mathbb{R}^{n  \times n}$, let $M^\text{d}$ denote the diagonal matrix with the diagonal elements of $M$, and $M^\text{nd} \triangleq M-M^\text{d}$ represent the off-diagonal entries. The ``Metzlerized" matrix\footnote{A Metzler matrix  is a matrix with nonnegative off-diagonal elements.} of $M$ is defined as $M^{\text{m}} \triangleq M^\text{d}+|M^\text{nd}|$. Moreover, its logarithmic and induced matrix norms corresponding to the $\infty$-norm are $\mu_\infty(M)=\max_i (M_{ii} +\sum_{j\neq i} |M_{ij}|)$ and $\|M\|_\infty=\max_i  \sum_{j} |M_{ij}|$, respectively. All inequalities involving vectors or matrices are interpreted element-wise, and the matrices of zeros and ones of dimension $n \times p$ are denoted as $\mathbf{0}_{n \times p}$ and $\mathbf{1}_{n \times p}$, respectively. 
 A function $\alpha: \mathbb{R}_+ \to \mathbb{R}_+$ is of class $\mathcal{K}$ if it is continuous, positive definite (i.e., $\alpha(x)=0$ for $x=0$; $\alpha(x)>0$ otherwise), and strictly increasing, and of class $\mathcal{K}_{\infty}$ if it is additionally unbounded. A function $\lambda : \mathbb{R}_+\times \mathbb{R}_+ \to \mathbb{R}_+$ is of class $\mathcal{KL}$ if, for every fixed  $t\geq 0$, $\lambda(\cdot,t)$ is of class $\mathcal{K}$ and for each fixed $s \geq 0$, $\lambda(s,t)$ decreases to zero as $t \to \infty$. Further, an interval $\mathcal{I}$ in $n$-dimensional space, defined as $[\underline{x}, \overline{x}] \subseteq \mathbb{R}^n$, comprises all vectors $x$ belonging to $\mathbb{R}^{n}$ satisfying $\underline{x} \leq x \leq \overline{x}$.

\vspace{-0.15cm}

\begin{defn}[Embedding System \cite{khajenejad2023tight}]\label{def:embedding}
For an uncertain linear system represented by $x^+=Ax+Ww$ with noise $w\in[\underline{w},\overline{w}]$ and initial state $x_0\in [\underline{x}_0,\overline{x}_0]$, its corresponding linear embedding system is a
$2n$-dimensional linear system with initial condition $\begin{bmatrix} \underline{x}_0^\top & \overline{x}_0^\top\end{bmatrix}^\top$, given as:
\begin{align} \label{eq:embedding}
\begin{bmatrix}{\underline{x}}_t^+ \\ {\overline{x}}_t^+ \end{bmatrix}=\begin{bmatrix}  A^{\uparrow}\underline{x}_t-A^{\downarrow}\overline{x}_t + W^{\oplus}\underline{w}_t-W^{\ominus}\overline{w}_t \\  A^{\uparrow}\overline{x}_t-A^{\downarrow}\underline{x}_t + W^{\oplus}\overline{w}_t-W^{\ominus}\underline{w}_t \end{bmatrix}.
\end{align}
where for the continuous-time case, $x^+_t=\dot{x}_t$, $A^{\uparrow}=A^d+A^{nd,\oplus}$ and $A^{\downarrow}=A^{nd,\ominus}$, while for the discrete-time case, $x^+_t=x_{t+1}$, $A^{\uparrow}=A^\oplus$ and $A^{\downarrow}=A^{\ominus}$.
\end{defn}
\vspace{-0.15cm}

Note that according to \cite[Proposition 3]{khajenejad2023tight}, the embedding system in \eqref{eq:embedding} has a \emph{state framer property}, i.e., its states are guaranteed to frame the unknown true states $x_t$: $\underline{x}_t\le x_t\le \overline{x}_t$, $\,\forall\,t\in\mathbb{T}$, for both discrete- and continuous-time systems.

Finally, one important property of the logarithmic and vector norms induced by the $\infty$-norm that we will later leverage for our observer design is:
\vspace{-0.2cm}
\begin{lem}\label{lem:normProp}
    Consider any $Q \in \mathbb{R}^{m \times m}$ and its Metzlerized and absolute value matrices, $Q^\textnormal{m}$ and $|Q|$, respectively. Then, 
    $$\mu_\infty(Q)=\mu_\infty(Q^\textnormal{m}), \quad \|Q\|_\infty=\||Q|\|_\infty.$$
\end{lem}
\vspace{-0.3cm}
\begin{proof}
    This result follows directly from the definitions of logarithmic and induced matrix norms corresponding to the $\infty$-norm as well as the definitions of the Metzlerization and absolute value operations.
\end{proof} 

\vspace{-0.35cm}
\section{Problem Formulation} \label{sec:Problem}
\vspace{-0.1cm}
We consider linear discrete-time (DT) or continuous-time (CT) systems with additive bounded uncertainties:
\vspace{-0.1cm}
\begin{align} 
\begin{array}{ll}
\mathcal{G}:
\begin{cases}
\begin{array}{rl}x_t^{+}\!\!\! &= Ax_t + Bu_t + Ww_t, \\
y_t \!\!\!&= Cx_t + Du_t + Vv_t,\end{array}
\label{eq:dynamics}
\end{cases} \quad \forall\;t \in \mathbb{T},
\end{array}
\end{align}
where, if $\mathcal{G}$ is a CT system, $x_t^+=\dot{x}_t, \mathbb{T} = \mathbb{R}_{\ge 0}$ and for DT $\mathcal{G}$, $x_t^+=x_{t+1}, \mathbb{T}= \{0\}\cup \mathbb{N}$. Additionally, the system state is $ x_t \in \mathcal{X} \subset \mathbb{R}^m $, the measured output $ y_t \in \mathbb{R}^l $, the control input $ u_t \in \mathbb{R}^s $,  the process noise $w_t \in \mathcal{W} \triangleq [\underline{w},\overline{w}] \subset \mathbb{R}^{n_w}$, and the measurement noise $v_t \in \mathcal{V} \triangleq [\underline{v},\overline{v}] \subset \mathbb{R}^{n_v}$, with an uncertain initial state $x_0\in\mathcal{X}_0=[\underline{x}_0,\overline{x}_0]$. Further, we assume that the input $u_t$ (if applicable) and output $y_t$ for all $t \in \mathbb{T}$, and the uncertainty bounds $\underline{w}$, $\overline{w}$, $\underline{v}$, $\overline{v}$, $\underline{x}_0$, and $\overline{x}_0$ are  known, and that the pair $(A,\;C)$ is detectable. Without loss of generality, we can also rescale $V$ and $W$ to ensure $\overline{v} = -\underline{v} = \mathbf{1}_{n_v \times 1}$ and  $\overline{w} = -\underline{w} = \mathbf{1}_{n_w \times 1}$, which can be useful to ensure their effects are uniformly penalized.

The problem of interest can be stated as:\vspace{-0.2cm} 
\begin{problem}\label{prob:PIO}
Given a detectable  uncertain CT/DT linear system in \eqref{eq:dynamics}, design  (optimal) time-invariant polytopic and interval observers with set estimates $\hat{X}_t$  for all $t \in \mathbb{T}$ that (i) enclose/contain the true state $x_t$, i.e., $x_t \in \hat{X}_t,\, \forall t\in \mathbb{T}$ (correctness property) and (ii) 
are input-to-state stable (ISS) with respect to their volumes $\textnormal{Vol}(\mathcal{X}_t)$, i.e.,  there exist class $\mathcal{KL}$ and $\mathcal{K}$ functions $\beta$ and $\gamma$ such that $\forall t \in \mathbb{T}$,
\begin{align}\label{eq:ISS}
\textnormal{Vol}(X_t) \leq \beta(\textnormal{Vol}(\mathcal{X}_0),t)+\gamma(\textnormal{Vol}(\mathcal{W})+\textnormal{Vol}(\mathcal{V})). 
\end{align}
\end{problem}

\vspace{-0.25cm}
\section{Proposed Polytopic and Interval Observers} \label{sec:observer}\vspace{-0.05cm}

First, since the pair $(A,\;C)$ is detectable by assumption, we can construct an equivalent `closed-loop' system for the system given in \eqref{eq:dynamics} by adding a zero term $L(y_t-Cx_t-Du_t-Vv_t)$ to the state equation and rearranging, to obtain:
\begin{align}
\!\mathcal{G}_{eq}:\!
\begin{cases}\!\begin{array}{rl}
x_t^{+}
\!\!&=A_{cl}x_t \!+\! Ly_t \!+\! Ww_t\!-\!LVv_t\!+\!(B\!-\!LD)u_t,\\
y_t \!\!&=C x_t+Du_t+V v_t,
\end{array}\end{cases}\hspace{-0.6cm}
\label{eq:dynamics_cl}
\end{align}
with any $L$ such that $A_{cl}=A-LC$ is Hurwitz (CT) or Schur (DT) stable. In particular, $L$ can be designed using optimal estimators such as $\mathcal{H}_\infty$ observers to minimize the $\mathcal{H}_\infty$ system gain of an estimation error system $\tilde{x}^+=(A-LC)\tilde{x}+W w_t -LV v_t$. 
However, it is well-known that interval observer designs for Hurwitz/Schur stable systems as in \eqref{eq:dynamics_cl} remain challenging and often require time-varying 
observers that can be complicated to implement, prone to numerical errors, and conservative. Hence, in the following Section \ref{sec:Non-Square}, we introduce a new time-invariant coordinate transformation that is potentially non-square. We then show in Section \ref{sec:obsv} that polytopic and interval observers can be designed for the transformed coordinates by leveraging mixed-monotone decompositions, which are correct and ISS without incurring any conservatism, thus solving Problem \ref{prob:PIO}.


\vspace{-0.2cm}
\subsection{Non-Square Time-Invariant Transformation}\label{sec:Non-Square}
\vspace{-0.1cm}

To overcome the complexity and numerical difficulties of time-varying observers, as well as to resolve a long-standing question of the existence of time-invariant observers for linear detectable CT and DT systems, this paper introduces a (non-square) time-invariant coordinate transformation that is inspired by results on infinity norm-based polyhedral Lyapunov functions for linear systems in \cite{polanski1995infinity,molchanov1986liapunov,christophersen2007further}: \vspace{-0.2cm}

\begin{thm}\label{thm:transformation}
Let $A_{cl} \in \mathbb{R}^{n \times n}$. A linear CT or DT system $ {x}^+_t = A_{cl} x_t $ is asymptotically stable if and only if there exist real matrices $ P \in \mathbb{R}^{m \times n} $, with full rank  (i.e., $\text{rank}(P) = n $, $m \geq n $), and $ Q \in \mathbb{R}^{m \times m} $ such that:
\vspace{-0.1cm}
\begin{subequations}
\begin{gather}
    PA_{cl} = QP \textnormal{ and } \mu_{\infty}(Q) < 0 \quad \textnormal{(CT)},\\
    PA_{cl} = QP \textnormal{ and } \|Q\|_{\infty} < 1 \quad \textnormal{(DT)},
\end{gather}
\end{subequations}
where $ \mu_{\infty}(Q) $ and \( \|Q\|_{\infty} \) are the logarithmic norm induced by the vector $ \infty $-norm \cite{Torsten1975OnLogarithmic} and the induced matrix $ \infty $-norm, respectively. Additionally, the function $V(x) = \|Px\|_{\infty}$ is a (polyhedral) Lyapunov function, satisfying $\dot{V}(x_t) < 0 $ (CT) or $ V(x_{t+1}) < V(x_t) $ (DT) for all $ x_t \neq 0 $. 
\end{thm}
\vspace{-0.35cm}

\begin{proof}
The sufficiency for CT and DT systems was proven by various authors, e.g., \cite{polanski1995infinity,molchanov1986liapunov}, while the necessity is proven by the explicit construction of real $P$ and $Q$ in \cite[Section 3]{polanski1995infinity} and \cite[Section 4]{christophersen2007further}  for CT and DT systems.
\end{proof}\vspace{-0.15cm}

Specifically, we leverage the real $P\in \mathbb{R}^{m \times n}$ and $Q\in\mathbb{R}^{m \times m}$  matrices corresponding to $A_{cl}$ in \eqref{eq:dynamics_cl} from the above theorem to propose the new transformed coordinate as
\begin{align}
    z_t = P x_t,
\end{align}
where since $P$ has full (column) rank by construction, $x_t$ can be recovered using the Moore-Penrose pseudoinverse of $P$, i.e., $ P^\dagger z_t = P^\dagger P x_t = x_t$. Then, we can obtain the transformed CT or DT system corresponding to \eqref{eq:dynamics_cl}:
\begin{align}
\begin{array}{rl}
z_{t}^{+} \hspace{-0.1cm}=\! Qz_{t} \hspace{-0.055cm}+\hspace{-0.055cm} PLy_t \hspace{-0.055cm}+\hspace{-0.055cm} PWw_t \hspace{-0.055cm}-\hspace{-0.055cm} PLVv_t \hspace{-0.055cm}+\hspace{-0.055cm} P(B 
\hspace{-0.055cm}-\hspace{-0.055cm} LD)u_t,
\label{eq:lifted_dynamics}
\end{array}
\end{align}
where we applied $PA_{cl}=QP$ from Theorem \ref{thm:stability}.

Note that unlike time-invariant or time-varying coordinate transformations for interval observers in the literature, e.g., \cite{chambon2016overview,mazenc2011interval,mazenc2014interval}, the above transformation does not attempt to find $Q$ that is Metzler for CT systems or non-negative for DT systems. In fact, the constructions provided in \cite[Section 3]{polanski1995infinity} for CT systems and in \cite[Section 4]{christophersen2007further}  for DT systems result in $Q$ that is \emph{not} Metzler (cf. \eqref{eq:Q_CT}) \emph{nor} non-negative (cf. \eqref{eq:Q_DT}). Moreover, the resulting \emph{stability} properties of dynamical systems with $Q$ as the system matrix, i.e., for $z^+=Qz$, are also enhanced. In particular, it has been shown  in \cite[Lemma 1c]{Torsten1975OnLogarithmic} that $\alpha(Q) \le \mu_\infty(Q)$ and in \cite[Lemma 5.6.10]{horn2012matrix} that $\rho(Q) \le \|Q\|_\infty $, where $\alpha(Q) = \max \{\textnormal{Re} (\lambda) \mid \lambda \in \Lambda(Q)\}$ and $\rho(Q) = \max \{|\lambda| \mid \lambda \in \Lambda(Q)\}$ are the  spectral abscissa and spectral radius of $Q$, respectively, with $\Lambda(Q)$ being the spectrum (set of eigenvalues) of $Q$. Hence, $\mu_\infty(Q) < 0$ and $\|Q\|_\infty< 1$ are ``stronger" properties for asymptotic stability of the $z^+=Q z$ system in CT and DT, respectively, than the more conventional eigenvalue-based properties of $\alpha(Q)<0$ (CT) and $\rho(Q)<1$ (DT).

Further, it is worth noting that the dimension of $z_t$ is $m \ge n$, i.e., our proposed transformation may be non-square leading to more (lifted) states. Specifically, for an $A_{cl}$ with $n_r$ real eigenvalues and $n_c=\frac{1}{2}(n-n_r)$ complex conjugate eigenvalue pairs of the form $\sigma_i \pm j \omega_i, i=1,\hdots,n_c$,  the constructions of $P$ and $Q$ in \cite[Section 3]{polanski1995infinity} for CT systems and in \cite[Section 4]{christophersen2007further}  for DT systems lead to $m = n_r + \sum_{i=1}^{n_c} m_i$, with  
$m_{i}$ being the smallest constant $c_i$ (with $c_i \ge 2$) for each complex conjugate eigenvalue pair $\sigma_i \pm j \omega_i$ such that:
\begin{subequations}\label{eq:m_i}
\begin{gather}
\frac{\sigma_i}{\omega_i} < \frac{\cos(\frac{\pi}{c_i}) - 1}{\sin(\frac{\pi}{c_i})} \quad \textnormal{(CT)},\\
\sqrt{\sigma_i^2 + \omega_i^2} < \frac{\cos(\frac{\pi}{2c_i})}{\cos(\psi_i - \frac{2\pi}{c_i}(2\lfloor \frac{\psi_i c_i}{\pi} \rfloor+1 )} \quad \textnormal{(DT)},
\end{gather} 
\end{subequations}
where $\lfloor \cdot \rfloor$ is the floor operator and 
$\psi_i = \begin{cases} 
\frac{\pi}{2}, & \text{if } \sigma_i = 0, \\
\tan^{-1} \frac{|\omega_i|}{|\sigma_i|}, & \text{otherwise}.
\end{cases}$ 
However, note that we can also choose any $m_{i}$ that is larger than the minimum $c_i$ and in fact, it generally leads to smaller volumes of the set estimates but at the cost of having more states. 
Further, from \eqref{eq:m_i}, we observe that $P \in \mathbb{R}^{n \times n}$ is square and time-invariant, i.e., $m=n$, if for all $i=1,\hdots,n_c$,
\begin{subequations}
\begin{gather}
\sigma_i < -\omega_i  \quad \textnormal{(CT)},\\
\sqrt{\sigma_i^2 \!+\! \omega_i^2}(\cos(\psi_i \!-\! \pi(2\lfloor \frac{2 \psi_i }{\pi} \rfloor\!+\!1 )) < \cos(\frac{\pi}{4}) \ \textnormal{(DT)},
\end{gather} 
\end{subequations}
which, on its own, contributes to answering the existence question of an interval observer with a (square) time-invariant similarity transformation. As a corollary, the above conditions imply that time-invariant polytopic and interval observers always exist for any detectable systems whose (closed-loop) system matrix $A_{cl}$ only have real eigenvalues.

Finally, while the time-varying approaches in \cite{mazenc2011interval,mazenc2014interval,meslem_Hinf_2020,meslem2021reachability} appear to only have $n$ states, if the time-varying matrices are also treated as states, then they can be viewed as having $n^2+n$ states, whereas the proposed time-invariant transformed system only has $m$ states, where we often have $n \le m < n^2+n$.

For the sake of readability and completeness, the  construction methods of $P$ and $Q$ for CT and DT systems from \cite[Section 3]{polanski1995infinity}  and  \cite[Section 4]{christophersen2007further} that rely on real Jordan canonical forms are recapped in the Appendix.

\subsection{Polytopic and Interval Observer Design}\label{sec:obsv}
Next, building upon the transformed dynamics in \eqref{eq:lifted_dynamics} using the transformation matrices $P$ and $Q$ computed by applying Theorem \ref{thm:transformation} to $\mathcal{G}_{eq}$ in  \eqref{eq:dynamics_cl} with a Hurwitz/Schur stable $A_{cl}$, we introduce our correct and stable polytopic and interval observers that are constructed as  the mixed-monotone embedding system in Definition \ref{def:embedding} for \eqref{eq:lifted_dynamics}:
\begin{align}\label{eq:observer}\small
\hat{\mathcal{G}}:
\begin{cases}\hspace{-0.1cm}
\begin{array}{rl}
\underline{z}^+_{t}&=Q^\uparrow\underline{z}_{t}\!-\!Q^\downarrow\overline{z}_{t}\!+\!(PW)^{\oplus}\underline{w}\!-\!(PW)^{\ominus}\overline{w}\\
&\,+ (PLV)^{\ominus}\underline{v}\!-\! (PLV)^{\oplus}\overline{v}\!+\!PL y_{t}\!+\! P(B\!-\!LD)u_t,\\
\overline{z}^+_{t}&=Q^\uparrow_x\overline{z}_{t}\!-\!Q^\downarrow\underline{z}_{t}\!+\!(PW)^{\oplus}\overline{w}\!-\!(PW)^{\ominus}\underline{w}\\
&\,+(PLV)^{\ominus}\overline{v}\!-\!(PLV)^{\oplus}\underline{v}\!+\!PL y_{t}\!+\! P(B\!-\!LD) u_t,
\end{array}
\end{cases}\hspace{-0.3cm}
\end{align}
where $Q^\uparrow \triangleq Q^\text{d}+Q^{\text{nd},{\oplus}}$, $Q^\downarrow \triangleq Q^{\text{nd},{\ominus}}$, $\underline{z}^+_{t}=\dot{\underline{z}}_{t}$ and $\overline{z}^+_{t}=\dot{\overline{z}}_{t}$ if \eqref{eq:lifted_dynamics} is a linear CT system; and $Q^\uparrow \triangleq Q^{\oplus}$, $Q^\downarrow \triangleq Q^{\ominus}$, $\underline{z}^+_{t}=\underline{z}_{t+1}$ and $\overline{z}^+_{t}=\overline{z}_{t+1}$ if \eqref{eq:lifted_dynamics} is a linear DT system.

By the framer property of mixed-monotone embedding systems \cite[Proposition 3]{khajenejad2023tight}, the interval estimates $\overline{z}_t$ and $\underline{z}_t$ above frame the true transformed state $z_t=Px_t$, i.e., $\underline{z}_t\le z_t=Px_t\le\overline{z}_t$ for all $t\in \mathbb{T}$. Moreover, we can bound the output equation in \eqref{eq:dynamics_cl}, i.e., $Cx_t=y_t-Du_t-Vv_t$, by applying \cite[Lemma 1]{efimov2013interval} to obtain
\begin{gather*}
    y_t-Du_t -V^\oplus \overline{v} +V^\ominus \underline{v}
 \le C x_t \le y_t-Du_t - V^\oplus \underline{v} + V^\ominus \overline{v}.
\end{gather*}
Combining the above, our polytopic estimate is given by \begin{align}\label{eq:polytope}
\small\begin{array}{r}
    \hat{X}_t^P = \left\{ x \in \mathcal{X} \mid \begin{bmatrix}
        P\\ -P \\ C \\-C
    \end{bmatrix}x \!\leq\! \begin{bmatrix}\overline{z}_t\\ -\underline{z}_t \\ y_t\!-\!Du_t \!-\!V^\oplus \underline{v} \!+\!V^\ominus \overline{v}
    \\ -y_t\!+\!Du_t \!+\!V^\oplus \overline{v} \!-\!V^\ominus \underline{v}\end{bmatrix} \right\}.
    \end{array}
\end{align}

Further, if interval estimates are desired, we can outer-/over-approximate the polytopic estimate to obtain
\begin{gather} \label{eq:estimatesInt}
    \hat{X}_t^I = \{x \in \mathcal{X}\mid \underline{x}_t \le x \le \overline{x}_t\},
\end{gather}
with upper and lower framers given by 
\begin{align}\label{eq:interval}
\begin{array}{rl}\small
\begin{bmatrix}
    \overline{x}_t\\ \underline{x}_t
\end{bmatrix} = & \begin{bmatrix}
 (\tilde{P}^{\dagger})^\oplus \\  -(\tilde{P}^{\dagger})^\ominus   
\end{bmatrix}\begin{bmatrix}
    \overline{z}_t\\ y_t-Du_t - V^\oplus \underline{v} + V^\ominus \overline{v}
\end{bmatrix}\\
&+\begin{bmatrix}
 -(\tilde{P}^{\dagger})^\ominus \\  (\tilde{P}^{\dagger})^\oplus   
\end{bmatrix}\begin{bmatrix}
    \overline{z}_t\\ y_t-Du_t -V^\oplus \overline{v} +V^\ominus
\end{bmatrix},
    \end{array}
\end{align}
where $\tilde{P}\triangleq \begin{bmatrix}
    P^\top & C^\top
\end{bmatrix}^\top$ has full column rank by construction of $P$ in Theorem \ref{thm:transformation} and $\tilde{P}^\dagger$ is its Moore-Penrose pseudoinverse satisfying $\tilde{P}^\dagger\tilde{P}=I$. Hence, the above framers can be derived by applying \cite[Lemma 1]{efimov2013interval} to $x_t=\tilde{P}^\dagger\tilde{P}x_t$ with the inequalities for $\tilde{P}x_t$ given in \eqref{eq:polytope}.

From the above derivation with the application of \cite[Proposition 3]{khajenejad2023tight} and \cite[Lemma 1]{efimov2013interval}, the  correctness property of our proposed polytopic and interval estimates (solving Problem \ref{prob:PIO}(i)) follows immediately, as summarized below.
\vspace{-0.2cm}

\begin{thm}[{
Correctness}]\label{lem:correctness}
Given a detectable linear CT or DT system $\mathcal{G}$ and its equivalent closed-loop system in \eqref{eq:dynamics_cl} with any $L$ that stabilizes $A_{cl}=A-LC$, 
the polytopic estimate $\hat{X}_t^P$ in \eqref{eq:polytope} and the interval estimate in \eqref{eq:estimatesInt}--\eqref{eq:interval}, obtained via the mixed-monotone embedding system  $\hat{\mathcal{G}}$ in \eqref{eq:observer} is a \emph{correct} enclosure of the true state $x_t$, i.e.,  $x_t \in \hat{X}_t^P \subseteq \hat{X}^I_t$ for all $t\in \mathbb{T}$ and for all $w_t \in \mathcal{W}$, $v_t \in \mathcal{V}$.
\end{thm}\vspace{-0.2cm}

Next, we prove that the proposed observers solve Problem \ref{prob:PIO}(ii), i.e., their set volumes satisfies the input-to-state stability (ISS) property in \eqref{eq:ISS}. \vspace{-0.2cm}

\begin{thm}[Input-to-State Stability]\label{thm:stability}
Consider a detectable linear CT or DT system $\mathcal{G}$ in \eqref{eq:dynamics}, any $L$ that stabilizes $A_{cl}=A-LC$, and the resulting equivalent closed-loop system in \eqref{eq:dynamics_cl}. The proposed polytopic/interval observer $\hat{\mathcal{G}}$ in \eqref{eq:observer} with polytopic and interval estimates in \eqref{eq:polytope}--\eqref{eq:interval} is input-to-state stable (ISS) with respect to their volumes $\hat{X}_t \in \{\hat{X}_t^P,\hat{X}_t^I\}$, i.e., there exist class $\mathcal{KL}$ and $\mathcal{K}$ functions $\beta$ and $\gamma$ such that $\forall t \in \mathbb{T}$, \eqref{eq:ISS} holds.
\end{thm}

\begin{proof}
First, we define the framer error $\varepsilon_t\triangleq\overline{z}_{t}-\underline{z}_{t}$ for \eqref{eq:observer} and write its resulting error system as:
\begin{align}\label{eq:observer_error}
\tilde{\mathcal{G}}:
\begin{array}{rl}
\varepsilon^+_{t}&=\tilde{Q}\varepsilon_{t}\!+\!f_\varepsilon,
\end{array}
\end{align}
with $\tilde{Q}= Q^d+|Q^{nd}|=Q^m$ and $\varepsilon^+_{t}=\dot{\varepsilon}_{t}$ if \eqref{eq:lifted_dynamics} is a CT system, and $\tilde{Q}= |Q|$ and $\varepsilon^+_{t}=\varepsilon_{t+1}$ if \eqref{eq:lifted_dynamics} is a DT system, while $f_\varepsilon\triangleq |PW|\Delta{w}\!+ |PLV|\Delta{v}$, $\Delta{w}\triangleq\overline{w}\!-\!\underline{w}$ and $\Delta{v}\triangleq\overline{v}\!-\!\underline{v}$.

    By Theorem \ref{thm:transformation}, the  $Q$ matrices for CT and DT systems satisfy $\mu_\infty(Q)<0$ and $\|Q\|_\infty <1$, respectively. And further, by Lemma \eqref{lem:normProp}, $\mu_{\infty}(Q^m)<0$ and $\||Q|\|_{\infty}<1$. 
    
    Then, for CT systems, by applying \cite[Lemma 2]{Torsten1975OnLogarithmic} to \eqref{eq:observer_error}, we obtain 
    $\frac{d}{dt}(\norm{\varepsilon_t}_\infty)\le \mu_\infty({Q}^m)\norm{\varepsilon_t}_\infty + \norm{f_\varepsilon}_\infty$, whose solution satisfies
    \begin{align*} \norm{\varepsilon_t}_\infty &\textstyle\le e^{\mu_\infty({Q}^m)t}\norm{\varepsilon_0}_\infty+\int_0^t e^{\mu_\infty({Q}^m)(t-s)}ds\norm{f_\varepsilon}_\infty\\
        &\textstyle = e^{\mu_\infty({Q}^m)t}\norm{\varepsilon_0}_\infty+\frac{e^{\mu_\infty({Q}^m)t}-1}{\mu_\infty({Q}^m)}\norm{f_\varepsilon}_\infty.
         \end{align*}
         Since $\mu_{\infty}(Q^m)<0$, the CT framer error system is ISS.

         On the other hand, the  solution of the DT framer error system in \eqref{eq:observer_error} is given by $\varepsilon_t = |Q|^t\varepsilon_0 +\sum_{i=0}^{t-1-i}|Q|^i f_\varepsilon$, whose bounds can be found using triangle inequality: \begin{align*}
        \|\varepsilon_t\|_\infty &\textstyle \le \||Q|\|_\infty^t \|\varepsilon_0\|_\infty +\sum_{i=0}^{t-1}\||Q|\|_\infty^{i}\|f_\varepsilon\|_\infty\\
        &\textstyle=\||Q|\|_\infty^t \|\varepsilon_0\|_\infty + \frac{1-\||Q|\|_\infty^{t}}{1-\||Q|\|_\infty}\|f_\varepsilon\|_\infty.
         \end{align*}
Since $\||Q|\|_\infty<1$, the DT framer error system is also ISS. Finally, since  volumes of $\hat{X}^P_t$ and $\hat{X}^I_t$ are monotone with respect to $\|\varepsilon_t\|_\infty$, the theorem holds. 
\end{proof}

Importantly, we have also shown above that the proposed approach via the mixed-monotone embedding system $\hat{\mathcal{G}}$ in \eqref{eq:lifted_dynamics} does not impose any conservatism with respect to the ISS property, due to the ``enhanced'' stability properties of $Q$ discussed in Section \ref{sec:Non-Square}, enabled by Lemma \ref{lem:normProp}, which means that the approach applies to all detectable linear CT and DT systems. Furthermore, this establishes the fact that (lifted) time-invariant polytopic and interval observers always exist for all detectable linear CT and DT systems.

\section{Illustrative Examples}
In this section, the effectiveness of our approaches are demonstrated on three CT and DT systems 
in MATLAB using YALMIP \cite{YALMIP}, Gurobi \cite{gurobi}, and CORA \cite{Althoff2015ARCH} in comparison with several time-varying interval observers in the literature, specifically \cite{mazenc2011interval,meslem_Hinf_2020} (CT) and \cite{mazenc2014interval,meslem2021reachability} (DT). 

\vspace{-0.1cm}
\subsection{CT System Examples}\label{sec:CT_exm}

\begin{figure*}[t!]
    \centering
        \begin{subfigure}[h]{\textwidth}
        \centering
        \includegraphics[width=0.335\columnwidth,trim=3mm 0mm 3mm 0mm]{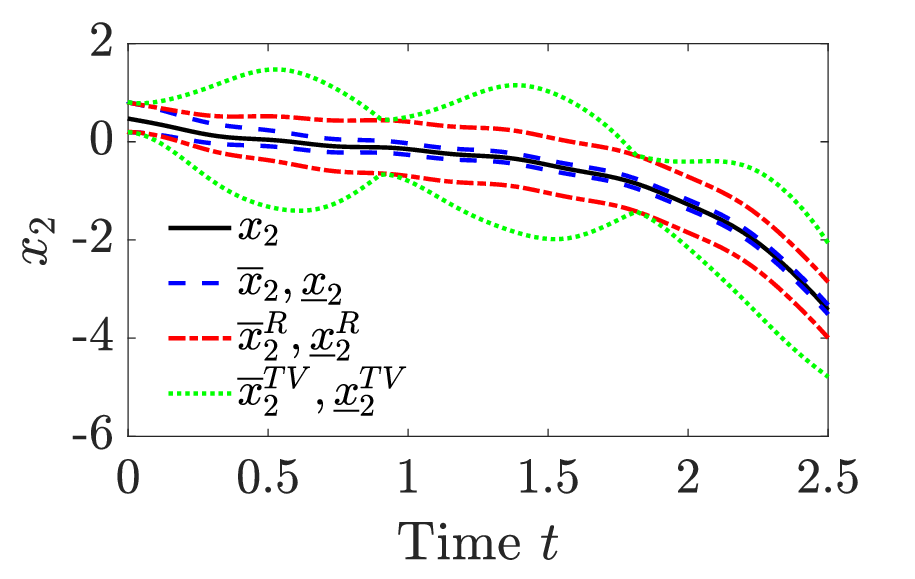}\
        \includegraphics[width=0.28\columnwidth,trim=3mm 3mm 3mm 0mm]{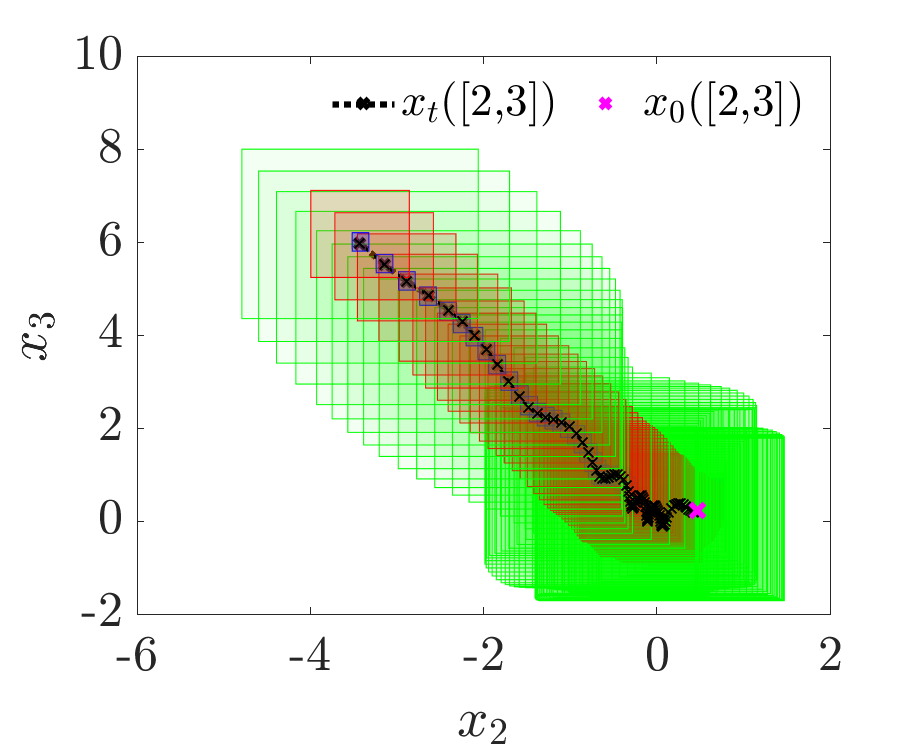}\
        \includegraphics[width=0.33\columnwidth,trim=3mm 3mm 3mm 0mm]{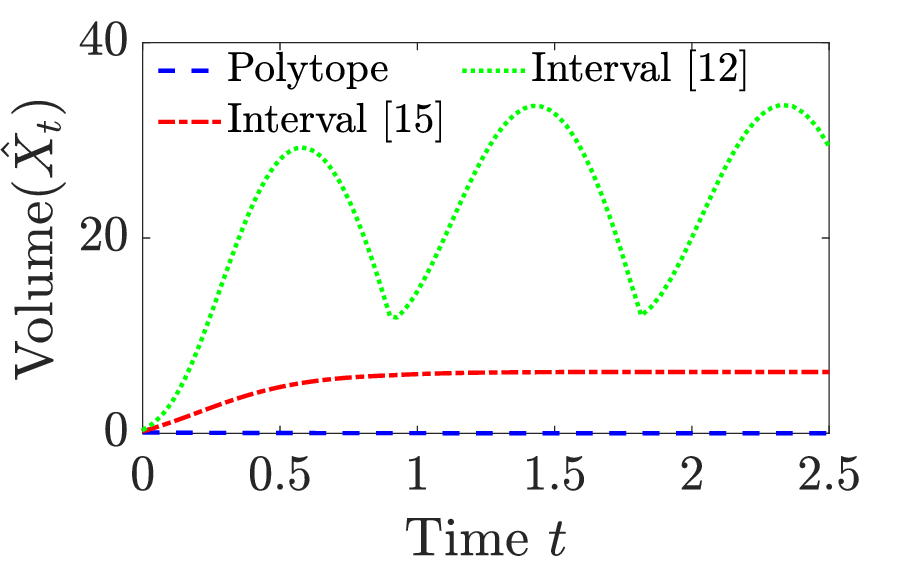}
        \caption{CT System from \cite[Sec. IV.A]{raissi2011interval}}
        \label{subfig:ct_system}
    \end{subfigure}\\
    \begin{subfigure}[h]{\textwidth}
        \centering
        \includegraphics[width=0.335\columnwidth,trim=3mm 0mm 3mm 0mm]{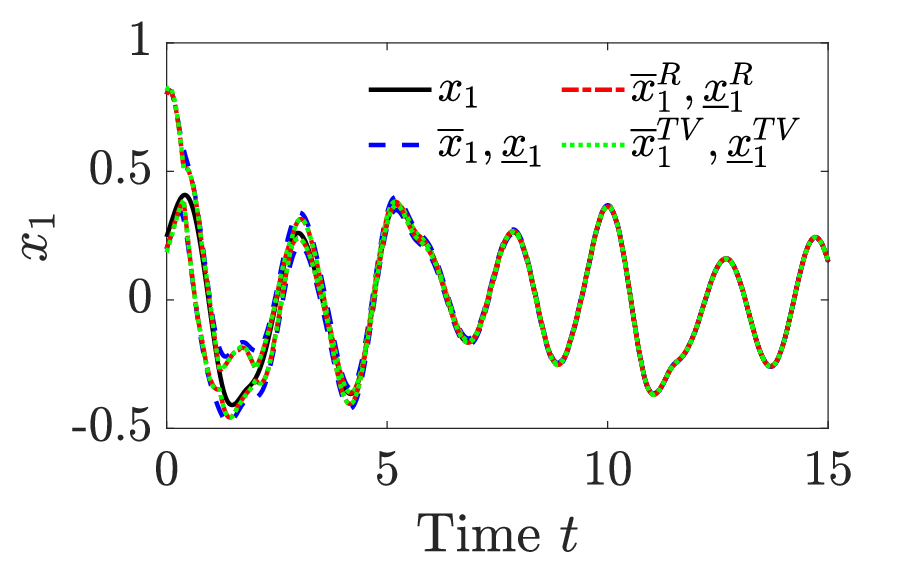}\
        \includegraphics[width=0.28\columnwidth,trim=3mm 3mm 3mm 0mm]{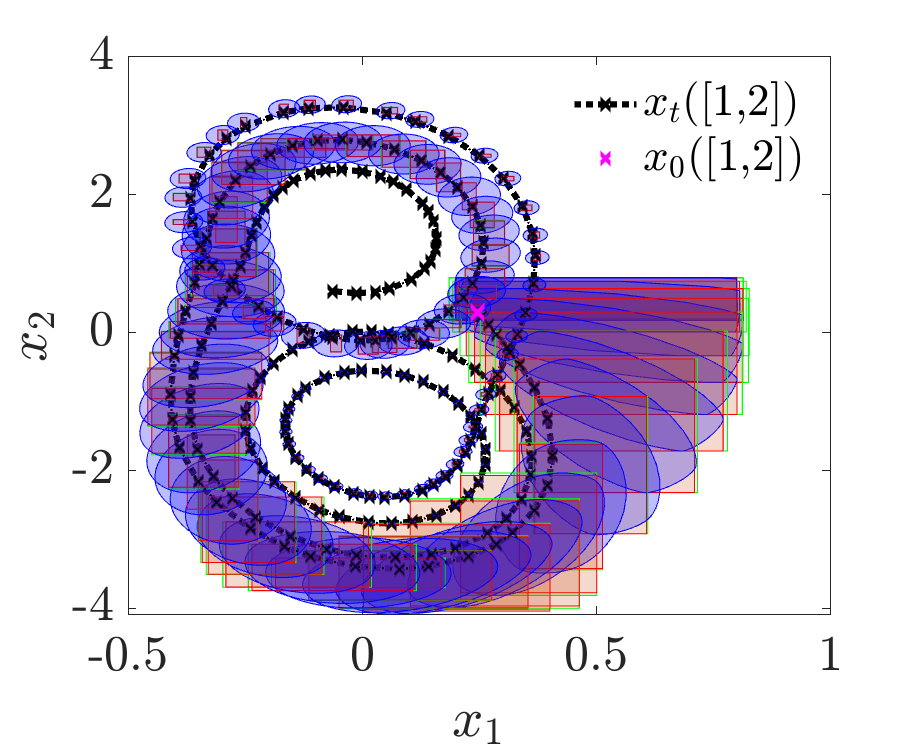}\
        \includegraphics[width=0.33\columnwidth,trim=3mm 3mm 3mm 0mm]{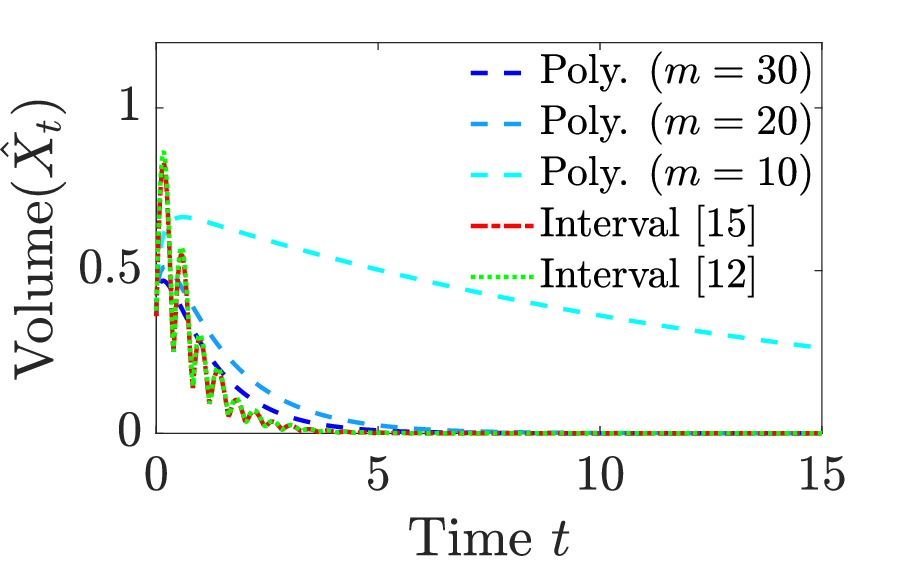}
        \caption{CT System from \cite[Eq. (18)]{Mazenc2010}}
        \label{subfig:ct_system2}
    \end{subfigure}\\
    \begin{subfigure}[h]{\textwidth}
        \centering
    \includegraphics[width=0.335\columnwidth,trim=3mm 0mm 3mm 0mm]{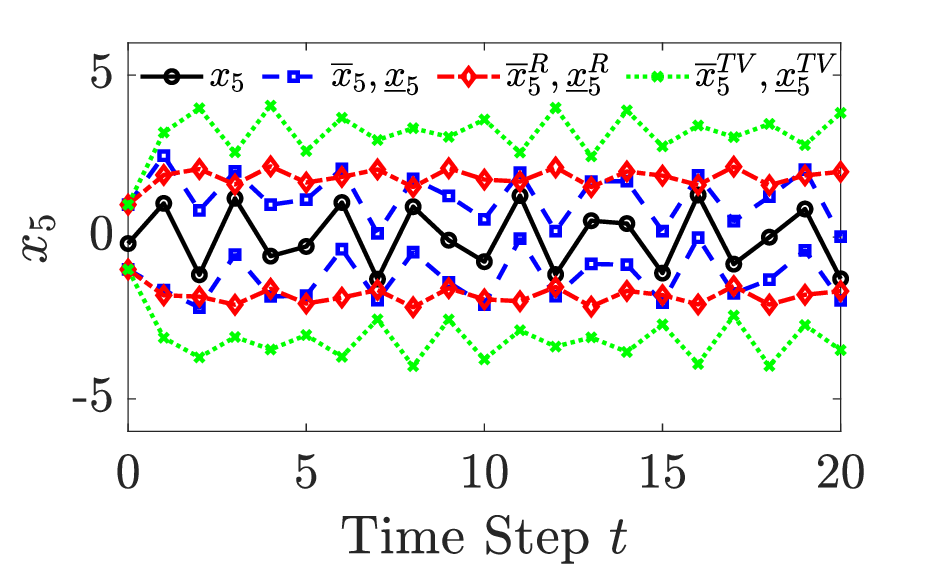}\
    \includegraphics[width=0.28\columnwidth,trim=3mm 3mm 3mm 0mm]{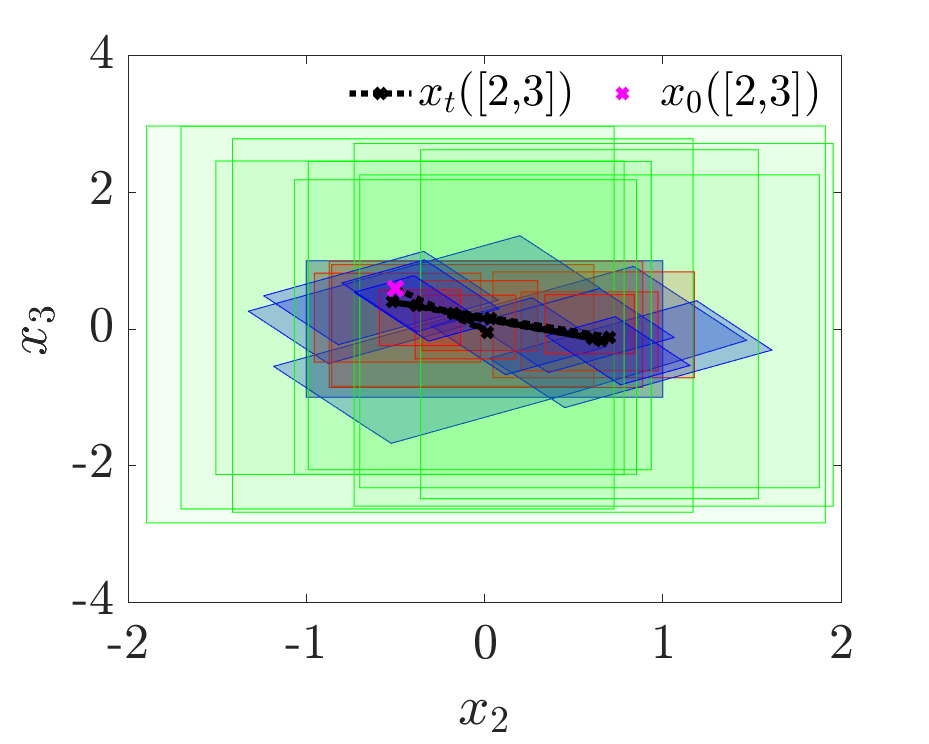}\
    \includegraphics[width=0.33\columnwidth,trim=3mm 3mm 3mm 0mm]{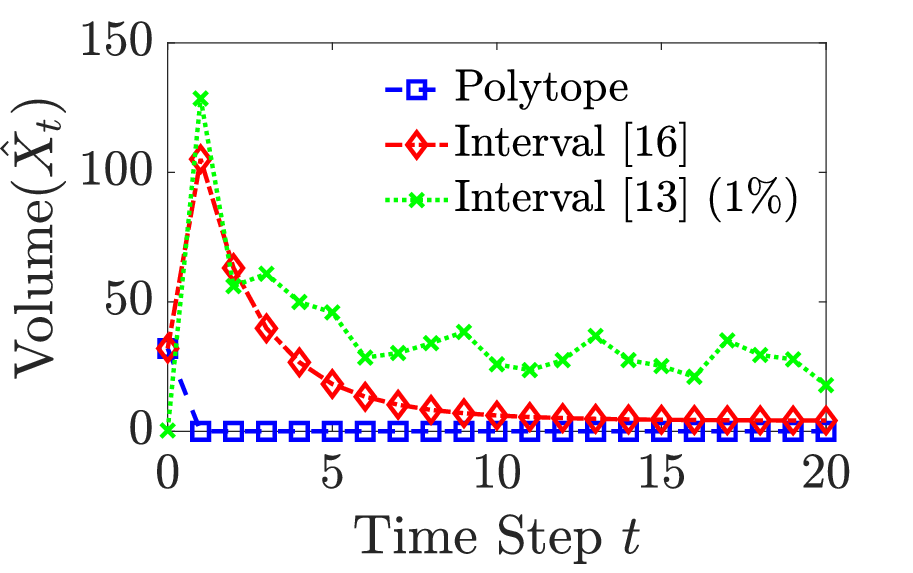}
        \caption{DT System from \cite[Sec. 4.2]{meslem2021reachability}}
    \label{subfig:dt_system}
    \end{subfigure}%
    \caption{Left: Representative state and its upper and lower framers obtained as projections of proposed polytopic estimates, and 
    from time-varying interval observers based on reachability in \cite{meslem_Hinf_2020} (CT) and \cite{meslem2021reachability} (DT), and on time-varying coordinate transformations in \cite{mazenc2011interval} (CT) and \cite{mazenc2014interval} (DT); Middle and right: Proposed polytopic sets in blue,  interval estimates of \cite{meslem_Hinf_2020} (CT) and \cite{meslem2021reachability} (DT) in red and interval estimates of \cite{mazenc2011interval} (CT) and \cite{mazenc2014interval} (DT) in green, and their respective volumes. Note that the (large) interval volumes from the approach in \cite{mazenc2014interval} (DT) are plotted with a 1/100th factor such that they are more comparable with the other methods.}\vspace{-0.6cm}
\end{figure*}

1) First, we consider the CT system in \cite[Sec. IV.A]{raissi2011interval} 
with system matrices  $A =
\begin{bmatrix}
2 & 0 & 0 \\
1 &-4 & \sqrt{3} \\
-1 & -\sqrt{3} & -4 
\end{bmatrix}$, 
$C =
\begin{bmatrix}
1 & 0 & 0
\end{bmatrix}$,  
$W =
\begin{bmatrix}
-10 & 0 & 3.4
\end{bmatrix}^\top$ and with empty $B$, $D$, and $V$. 
The initial state ${x}_0$ is  unknown but bounded within the interval $[0.2, 0.8]\times [0.2, 0.8]\times [0.2, 0.8]$, and the system is subject to an unknown but bounded process noise $w_t$ within $[-1,1]$, simulated using $w_t=\sin(15t)$. The observer gain $L=\begin{bmatrix}
    8.7827 & 0.5239 & -1.8195
\end{bmatrix}^\top$ is designed by using an $\mathcal{H}_\infty$ estimator design and the corresponding $P$ and $Q$ matrices of our proposed observer are found to be $P=\begin{bmatrix}
    -0.289 &      0    &     0 \\
    0.289 &  0  &  1\\
    0.0088 &  -1 &    0
\end{bmatrix}$ and $Q=\begin{bmatrix}
    -6.7827  &       0    &     0 \\
         0  & -4  &  1.732 \\
         0  & -1.732 &   -4
\end{bmatrix}$. Note that this corresponds to a time-invariant similarity transformation (i.e., with $m=n=3$), where $Q$ is notably not Metzler), yet our approach can yield input-to-state stable time-invariant interval and polytopic observers, while all known existing time-invariant interval observers cannot.

Figure \ref{subfig:ct_system} depicts the results from our polytopic observer, as well as those from time-varying interval observers based on time-varying coordinate transformation  \cite{mazenc2011interval} and reachability analysis \cite{meslem_Hinf_2020}. The left plot illustrates the time evolution of a representative state $x_2$ (for the sake of brevity), where the true state $x_2$  is depicted with a black solid line while the upper and lower framers from our polytopic observer (projected onto $x_2$ using \eqref{eq:interval}), denoted as $\overline{x}_2$ and $\underline{x}_2$, are in blue dashed lines, and the framers of  the interval-based methods from \cite{meslem_Hinf_2020} are in red dashed-dotted lines ($\overline{x}^R_2$ and $\underline{x}^R_2$), and from \cite{mazenc2011interval} are in green dotted lines ($\overline{x}^{TV}_2$ and $\underline{x}^{TV}_2$). On the other hand, the middle plot depicts the projection of the set estimates of the various methods (our approach in blue, the approach in \cite{meslem_Hinf_2020} in red, and the approach in \cite{mazenc2011interval} in green) onto the ($x_2,x_3$) plane, while the right plot illustrates the time evolution of their volumes.

From the plots, we observe that while all methods return set estimates that enclose the true states (i.e., they are all \emph{correct}), our proposed approach yields the tightest sets with the smallest volumes, highlighting the superior performance of our approach. On the other hand, the reachability-based method in \cite{meslem_Hinf_2020} is generally  better than the time-varying coordinate transformation approach in \cite{mazenc2011interval}.

 2) Moreover, we consider the chaotic Chua system in \cite[Eq. (18)]{Mazenc2010} with $A = \begin{bmatrix}-1 & 1\\ -14.9 & -0.29\end{bmatrix}$, $B=W=\begin{bmatrix}
    1 & 0
\end{bmatrix}^\top$, and empty $C$, $D$ and $V$, as well as $u_t$ being a known input computed using \cite[Eq. (11)]{Mazenc2010} (details omitted for brevity). For this system, it was shown in \cite{Mazenc2010} that it does not admit a time-invariant similarity transformation (with $m=n$). However, since this system is stable (and hence, detectable), our proposed approach with the non-square transformation applies and is able to find the corresponding $P$ and $Q$ matrices that satisfy Theorem \ref{thm:transformation}, albeit with the minimum $m=c_i$ satisfying \eqref{eq:m_i} found to be 10, i.e., we now have a lifted system with more states $z_t$. This also means that we can design stable time-invariant polytopic and interval observers for this system. Nonetheless, from Figure \ref{subfig:ct_system2} (right), we observe that the convergence rates of the set estimate volumes are generally faster with larger $m$, \zhu{
at the cost of marginally higher computational time (CPU times for solving the observer ODEs in MATLAB on an M1 Pro MacBook are 1.3195, 2.2083, and 3.0306 seconds for $m=10, 20, 30$, respectively).}

Then, we compare our proposed approach (with $m=30$) to those from time-varying interval observers in  \cite{mazenc2011interval,meslem_Hinf_2020}, and the results are depicted in Figure \ref{subfig:ct_system2}. From the plots, we observe that all three methods have comparable performances, with the set estimate volumes from our proposed approach being slightly smaller initially but they converge a little slower than the time-varying interval observers that have near identical estimates and volumes.

\vspace{-0.1cm}
 \subsection{DT System Example}\label{sec:DT_exm}
 \vspace{-0.1cm}

Next, we consider the DT system from \cite[Eq. (43)]{meslem2021reachability} 
with: 
\begin{gather*}
    A =
    \begin{bmatrix}
        -0.54 & 0.45 & 0.36 & 0 & 0 \\
        0.63 & 0.45 & 0.18 & 0.36 & 0 \\
        0.09 & 0.45 & 0.27 & 0.09 & 0.18 \\
        0 & 0 & 0.25 & 0.25\sqrt{2} & -0.25\sqrt{2} \\
        0 & 0 & 0 & 0.25\sqrt{2} & -0.25\sqrt{2}
    \end{bmatrix},\\
    C =
    \begin{bmatrix}
        1 & 0 & 0 & 0 & 0 \\
        0 & 0 & 0 & 1 & 0
    \end{bmatrix}, \quad
    W =
    \begin{bmatrix}
        -1 & 0 & 0 & 0 & 1
    \end{bmatrix}^{\top},
\end{gather*} and empty $B$, $D$, and $V$. 
The true initial state is  ${x}_0 = \begin{bmatrix}-0.3 & -0.5 & 0.6 & 0.9 & -0.2\end{bmatrix}^{\top}$, which is unknown, and the initial state set is assumed be a 5-hyperball $[-1,1]^5$. 
The system is subject to an unknown but bounded disturbance $ w_t $ within $[-1,1]$, simulated using a sinusoidal signal $w_t =  \sin(15 t) $. We adopt the observer gain $L$ from \cite[Eq. 49]{meslem2021reachability} that is based on an $\mathcal{H}_\infty$ observer design, and the corresponding $P$ and $Q$ for our proposed observer are: 
\begin{gather*}
L =\begin{bmatrix}
-0.3218 & 0.5486 & 0.0756 & 0.1861 & -0.1631 \\
0.1516 & 0.1922 & 0.0996 & 0.1457 & 0.0113
\end{bmatrix}^\top,\\
\small P =\begin{bmatrix}
    0.0034  &  0.0953  &  0.0557  &  0.0286 &  -0.0001\\
   -0.2835 &  -0.0670  &  0.2871  &   0.3631 &  -0.1784 \\
   -0.2686 &  -0.6393  &  0.7787   & 0.2609  &  0.0932 \\
    0.2801 &  -0.0283 &  -0.3428 &  -0.3917  &  1.1785 \\
    0.9798  & -0.1739 &  -0.6796  &  0.5974 &  -0.0312
\end{bmatrix},\\
\small Q = \textnormal{diag}(0.7288,\begin{bmatrix}
             0.0946  &  0.0347  \\     
       -0.0347 &   0.0946 
\end{bmatrix},\begin{bmatrix}
      -0.2809 &   0.2811\\
            -0.2811 &  -0.2809
\end{bmatrix}),
\end{gather*}
which again corresponds to a time-invariant similarity transformation without state lifting, where notably $Q$ is not non-negative, as was required by almost all time-invariant interval observers. Our proposed polytopic and interval observer, however, does not need a non-negative $Q$ and can still yield input-to-state set-valued observers.

Figure \ref{subfig:dt_system} compares the performance of our polytopic observer with two time-varying interval observers based on time-varying transformations in \cite{mazenc2014interval} and reachability analysis \cite{meslem2021reachability}. Analogous to the CT case, the left plot illustrates the time evolution of a representative state $x_5$ (in black solid lines) and the upper and lower framers from our proposed method  ($\overline{x}_5$, $\underline{x}_5$ in blue dashed lines), the approach in \cite{meslem2021reachability} ($\overline{x}^R_2$, $\underline{x}^R_2$ in red dashed-dotted lines) and in \cite{mazenc2014interval} ($\overline{x}^{TV}_2$ and $\underline{x}^{TV}_2$ in green dotted lines). 
Moreover, the middle plot shows the  projection of the polytopic estimates (blue) and interval estimates (red \cite{meslem2021reachability}/green \cite{mazenc2014interval}) onto the ($x_2,x_3$) plane, while the right plot depicts their volumes, where we had to scale down the very large volumes of the approach in \cite{mazenc2014interval} by a factor of 1\% for visual clarity and better comparison with the other two approaches.

As in the CT example, the polytopic sets enclose the system trajectory (black) with significantly reduced conservatism, underscoring the improved performance of our approach in terms of the volumes of the set estimates.
The consistently lower volume of our polytopic sets (blue) demonstrates that our approach yields both tighter estimates and faster convergence when compared with existing time-varying interval observers.

\vspace{-0.1cm}
\section{Conclusion and Future Work}
\label{sec:conclusion}
This work presented a novel framework for designing polytopic and interval observers for detectable linear continuous-time and discrete-time systems subjected to additive bounded uncertainties. The proposed time-invariant polytopic and interval observers are inherently correct by design, i.e., the set estimates enclose the true state at all times, without 
any positivity assumptions or constraints, which is achieved through the use of mixed-monotone decomposition and embedding system techniques, and the volumes of the set estimates are also proven to be input-to-state stable. For this design, we introduced a potentially non-square (lifted) coordinate transformation technique that builds upon techniques for polyhedral Lyapunov functions, which enabled the design of polytopic and interval observers for \emph{all detectable linear systems}. Our paper also resolved an open question on the existence of stable time-invariant interval observers for any detectable system, answering it in the affirmative, although it may sometimes require lifted/more states. 
Finally, we demonstrated that the proposed polytopic and interval observers outperformed existing time-varying interval observer designs in terms of obtaining tighter/smaller set estimates. As future work, we aim to extend this framework to nonlinear continuous-time and discrete-time systems, as well as hybrid or switched systems with known or unknown inputs.
\vspace{-0.12cm}

\bibliographystyle{IEEEtran}
\bibliography{main}

\balance

\appendix \normalsize

This appendix recaps the   construction of $P$ and $Q$ for CT and DT systems from \cite[Section 3]{polanski1995infinity}  and  \cite[Section 4]{christophersen2007further} (using a common/modified set of notations).

\subsubsection*{Step 1: Real Jordan Canonical Form}
First, recall that any real matrix $A_{cl}$ admits a real Jordan canonical form \cite[Section 3.4.1]{horn2012matrix}, i.e., there exists an invertible matrix $T$ such that $T^{-1} A_{cl} T = J =\textnormal{diag}(J_{r,1},\dots, J_{r,n_r},J_{c,1} \dots, J_{c,n_r}) \in \mathbb{R}^{n \times n}$, where $\textnormal{diag}$ returns a block diagonal matrix of the input matrices, with the first $n_r$ matrices being associated with real eigenvalues $\lambda_{R,i}$ of multiplicity $\ell_{r,i}$ of $A_{cl}$ and the other $n_c$ are associated with the complex conjugate eigenvalues $\lambda_{c,i}=\sigma_i \pm \omega_i$ of multiplicity $\ell_{c,i}$ of $A_{cl}$. Therefore, $n=\sum_{i=1}^{n_r} \ell_{r,i}+ \sum_{i=1}^{n_c} 2 \ell_{c,i}$, and $J_{r,i} \in \mathbb{R}^{\ell_{r,i} \times \ell_{r,i}},\ i=1,\dots,n_r$ and $J_{c,i} \in \mathbb{R}^{2\ell_{c,i} \times 2\ell_{c,i}}, \ i=1,\dots,n_c$ are given by 

\vspace{-0.4cm}
\begin{align*}\small
    J_{r,i} &= \begin{bmatrix}
    \lambda_{r,i} & 1                  & \cdots & 0 \\
    0         & \lambda_{r,i}          & \cdots & 0 \\
    \vdots    & \vdots        & \ddots & \vdots \\
    0         & 0             & 0         & \lambda_{r,i}
    \end{bmatrix}\!, \, 
        J_{c,i} =\begin{bmatrix}
        \Lambda_{c,i}       & I_2                & \cdots & 0 \\
        0         & \Lambda_{c,i}              & \cdots & 0 \\
        \vdots    & \vdots        & \ddots & \vdots \\
        0         & 0             & 0      & \Lambda_{c,i}
        \end{bmatrix}\!,
    \end{align*}
with $\Lambda_{c,i}=\begin{bmatrix} \sigma_i & \omega_i \\ -\omega_i & \sigma_i \end{bmatrix}$ and $I_2 = \begin{bmatrix} 1 & 0 \\ 0 & 1 \end{bmatrix}$.

\subsubsection*{Step 2: Determination of $m_i$ and construction of $P_{m_i}$ and $Q_{m_i}$} Next, for each complex conjugate eigenvalues, $\lambda_{c,i}=\sigma_i \pm \omega_i$, $i=1,\dots,n_c$, choose $m_i$ as any $c_i\ge 2$ such that \eqref{eq:m_i} holds, e.g., the smallest $c_i$ (via line search) such that \eqref{eq:m_i} holds to obtain the smallest $\hat{\mathcal{G}}$ in \eqref{eq:observer}. Then, we construct the corresponding $P_{m_i}$ (with \text{rank}$(P_{m_{i}}) = 2$) and $Q_{m_i}$:
\begin{equation}\small
P_{m_{i}} = \begin{bmatrix}
1 & 0 \\
\cos\frac{\pi}{m_{i}} & \sin\frac{\pi}{m_{i}} \\
\vdots & \vdots \\
\cos\frac{(m_{i} - 1)\pi}{m_{i}} & \sin\frac{(m_{i} - 1)\pi}{m_{i}}
\end{bmatrix}\in \mathbb{R}^{m_i\times 2},
\end{equation}
\paragraph{CT Systems} Let $\xi_i \triangleq \sigma_i-\omega_i \cot(\frac{\pi}{m_i})$ and $\psi_i \triangleq \omega_i \csc(\frac{\pi}{m_i})$ and construct 
\begin{align} \label{eq:Q_CT}\small
Q_{m_{i}} = 
\begin{bmatrix}
\xi_i & \psi_i & 0 & \cdots & 0 \\
0 & \xi_i & \psi_i & \cdots & 0 \\
\vdots & \ddots & \ddots & \ddots & \vdots \\
0 & 0 & 0 & \xi_i & \psi_i \\
-\psi_i & 0 & 0 & 0 & \xi_i
\end{bmatrix} \in \mathbb{R}^{m_{i} \times m_{i}};
\end{align}
\paragraph{DT Systems} Solve the following linear program
\begin{align*}\small
    \gamma_i=&\min \eta_i^\top\mathbf{1}_{m_i}\\
    &\textnormal{s.t.} \begin{bmatrix} \sigma_i \\ \omega_i \end{bmatrix} = P_{m_{i}}^\top \zeta_{i}, -\eta_i \le \zeta_{i}  \leq \eta_i,  \eta_i^\top\mathbf{1}_{m_i} < 1,
\end{align*}
with \( \zeta_{i} = \begin{bmatrix} \zeta_{i,1} & \dots & \zeta_{i,m_{i}} \end{bmatrix}^\top \)and construct
\begin{equation} \label{eq:Q_DT}\small
Q_{m_{i}} = 
\begin{bmatrix}
\zeta_{i,1} & \zeta_{i,2} & \zeta_{i,3} &\cdots & \zeta_{i,m_{i}} \\
-\zeta_{i,m_{i}} & \zeta_{i,1} & \zeta_{i,2} & \cdots & \zeta_{i,m_{i}-1} \\
\vdots & \ddots & \ddots & \ddots & \vdots \\
-\zeta_{i,3} & -\zeta_{i,4} & -\zeta_{i,5} &\cdots & \zeta_{i,2} \\
-\zeta_{i,2} & -\zeta_{i,3} & -\zeta_{i,4} & \cdots & \zeta_{i,1}
\end{bmatrix} \in \mathbb{R}^{m_{i} \times m_{i}}.
\end{equation}

\subsubsection*{Step 3: Determine $h_{r,i}$, $h_{c,i}$ and construct $H_{r,i}$ and $H_{c,i}$} 
For each real eigenvalues, $\lambda_{r,i}$, choose any $h_{r,i}$ that satisfies
\begin{gather*}
    \begin{array}{c}\text{Re}(\lambda_{r,i}) + h_{r,i}  < 0 \quad \textnormal{(CT)},\\
    |\lambda_{r,i}|+ h_{r,i}  < 1 \quad \textnormal{(DT)},
    \end{array}
\end{gather*}
and construct $H_{r,i} = \textnormal{diag}(1,h_{r,i}^{-1},\dots,h_{r,i}^{1-\ell_{r,i}})$.
For each complex conjugate eigenvalues, $\lambda_{c,i}=\sigma_i \pm \omega_i$, $i=1,\dots,n_c$, choose any $h_{c,i}$ that satisfies
\begin{gather*}
    \begin{array}{c}\xi_i+\psi_i+ h_{c,i} < 0\quad \textnormal{(CT)},\\ 
    \gamma_i +h_{c,i}<1 \quad \textnormal{(DT)},
    \end{array}
\end{gather*} 
and construct $H_{c,i} = \textnormal{diag}(I_2,h_{c,i}^{-1}I_2,\dots,h_{c,i}^{1-\ell_{c,i}}I_2)$.

\subsubsection*{Step 4: Construct $P$} Then, we can construct $$P=\textnormal{diag}(P_r,P_c)T^{-1},$$ where $P_r =\textnormal{diag}(H_{r,1},\dots,H_{r,n_r})$ and $P_c =\textnormal{diag}(\tilde{P}_{m_1}H_{c,1},\dots,\tilde{P}_{c,m_{n_c}}H_{c,n_c})$ with $\tilde{P}_{m_i}\triangleq\textnormal{diag}(P_{m_i},\dots,P_{m_i})\in \mathbb{R}^{m_i\ell_{c,i}\times 2\ell_{c,i}}$.

\subsubsection*{Step 5: Construct $Q$} Finally, we can construct $$\small Q\!=\!\textnormal{diag}(Q_r,Q_c),$$ with $Q_r \!=\!\textnormal{diag}(Q_{r,1},\dots,Q_{r,n_r})$ and $Q_c \!=\!\textnormal{diag}(Q_{c,1},\dots,Q_{c,n_c})$, as well as $Q_{r,i}\in \mathbb{R}^{\ell_{r,i} \times \ell_{r,i}}$, $Q_{c,i}\in \mathbb{R}^{m_{i}\ell_{c,i} \times m_{i}\ell_{c,i}}$ given by
\begin{align*}\small
    \begin{array}{l}Q_{r,i}\!\!=\!\!  \begin{bmatrix}
\lambda_{r_{i}} & h_{r,i}  & \cdots & 0 \\
0 & \lambda_{r_{i}} & \cdots & 0\\
\vdots & \ddots & \ddots & \vdots\\
0 & \cdots & \lambda_{r_{i}} & h_{r,i}  \\
 \cdots & \cdots & 0 & \lambda_{r_{i}}
\end{bmatrix}\!,\end{array} \!\!  
   \begin{array}{l} Q_{c,i} \!=\!\! \begin{bmatrix}
Q_{m_{i}} & h_{c,i} I_{m_{i}} & \cdots & 0 \\
0 & Q_{m_{i}} & \cdots & 0\\
\vdots & \ddots & \ddots & \vdots\\
0 & \cdots & Q_{m_{i}} & h_{c,i} I_{m_{i}}  \\
 \cdots & \cdots & 0 & Q_{m_{i}}
\end{bmatrix}\!.\end{array}\hspace{-0.3cm} 
\end{align*}

\end{document}